\documentclass[11pt]{article}
\usepackage{amssymb,fullpage,graphicx,url,amsmath,amsthm}

\usepackage{tikz,url}
\usepackage{algorithmicx,algpseudocode}
\usepackage{comment}

\newcommand{\N}{{\mathbb N}}

\newcommand{\Z}{{\mathbb Z}}

\newtheorem{Theorem}{Theorem}

\newtheorem{Definition}{Definition}
\newtheorem{Example}{Example}
\newtheorem{Conjecture}{Conjecture}

\newtheorem{Remark}{Remark}

\date{March 28, 2018}

\begin{document}

\title{Does universal controllability of
physical systems prohibit thermodynamic cycles?}

\author{
Dominik Janzing \\
{\small Max Planck Institute for Intelligent Systems}\\
{\small Max-Planck-Ring 4}\\
{\small 72076 T\"ubingen,  Germany}\\
{\small Email: \texttt{dominik.janzing@tuebingen.mpg.de}}\\
\\
Pawel Wocjan \\
{ \small Department of Computer Science, University of Central Florida}\\
{\small 4328 Scorpius Street}\\
{\small Orlando, FL 32816, USA}\\
{\small Email: \texttt{wocjan@cs.ucf.edu}}
}

\maketitle





\abstract{
Here we study the thermodynamic cost of
computation and control 
using 'physically universal' cellular automata or Hamiltonians.
The latter were previously defined as systems  that admit the implementation of any desired transformation  on
a finite target region by first initializing the state of the surrounding and 
then letting the system evolve according to its autonomous dynamics.
This way, one obtains a model of control where each 
region can play both roles the controller or the system to be controlled. In physically universal systems every degree of freedom is indirectly accessible by
operating on the remaining degrees of freedom. 

In a nutshell, the thermodynamic cost of an operation is then given by the size of the 
region around the target region that needs to be initialized. In the meantime, physically universal CAs have been constructed by Schaeffer (in two dimensions) and 
Salo \& T\"orm\"a (in one dimension). Here we show
that in Schaeffer's CA the cost for implementing
$n$ operations grows linearly in $n$, while operating in a thermodynamic cycle requires
sublinear growth to ensure zero cost {\it per operation} in the limit $n\to\infty$.
Although this particular result need not hold for general physically universal CAs, this strong notion of universality does imply a certain kind of instability of information, which could result in lower bounds on the cost of protecting information from its noisy environment.

The technical results of the paper are sparse and quite simple.
The contribution of the paper is mainly conceptual and consists in illustrating the type of thermodynamic questions raised by models of control that rely on the concept of physical universality. 
}

\section{Why thermodynamics of computation and control requires new models}

\subsection{The debate on thermodynamics of computation since the 1960s}

The question of whether there are fundamental lower bounds on the energy consumption of computing devices has attracted the attention of researchers since the
1960s. Landauer \cite{Landauer:61} realized that logically irreversible operations like erasure of memory space 
necessarily require to transfer the energy $\ln 2 k T$
per bit to the environment (with $k$ denoting Boltzmann's constant and
$T$ the temperature of the environment) due to the second law of thermodynamics.\footnote{In \cite{JWZGB} we have
argued that the energy requirements for {\it reliable} erasure 
are even larger than Landauer's bound 
when the state of the energy source is noisy, for instance if it is given by   two thermodynamic reservoirs of different temperatures. For further different perspectives on Landauer's principle see, e.g.,  \cite{Maroney2009,Sagawa2014,Wolpert2015}.}
Bennett \cite{Bennett:73} clarified that
computation can be performed without logically irreversible operations and thus Landauer's argument does not prove any fundamental lower bound for the
energy needed by computation tasks without further specification.
Ref.~\cite{viva2002} argues that physical models of reversible computation should include the clocking mechanism (that control the implementation of logical gates)
because otherwise one neglects the question of  how to implement
clocking in a thermodynamically reversible way (after all, if both gates and clocking device 
are described as quantum systems then the influence
of the latter on the former would, to some extent, also
imply an influence of the former on the latter \cite{controlled_controller}). 

\subsection{External clocking and control signals as loopholes}

To motivate this work step by step
we first  discuss the  thermodynamics of clocking and synchronization briefly which is a sophisticated problem \cite{Steudel,clockentropy,SynchrEntropy,Erker16}. Ref.~\cite{SynchrEntropy}, for instance, studies
some synchronization protocols that  suggest that thermodynamically reversible synchronization 
requires to exchange {\it quantum} information, which links the a priori different tasks of
reversible computation and quantum computing.\footnote{Here, the formal distinction between
quantum and classical clock signals as well as the conversion of time information between them is based on the rather general framework introduced in \cite{clock}.} 

Going beyond the question of whether implementing reversible logical operations
is possible in a thermodynamically reversible way, we ask whether implementing
unitary operations on some quantum system is possible in a thermodynamically reversible way.
Regardless of how we call the physical devices controlling the implementation
(we called it `clock' in the case of computation processes), also the implementation of a unitary $U$  requires to `change Hamiltonians' -- except for the special case where
$U=e^{-iHt}$ with $H$ being the free Hamiltonian of the system of consideration. 
However, do we really have appropriate models
for discussing the thermodynamic cost of `changing a system's Hamiltonian'? After all, describing a control field in classical terms is only a valid approximation if it can be considered {\it macroscopic}. 
For instance, a `macroscopic' number of electrons, 
sufficiently distant from some probe particle under consideration, could create such a `classical' field. 
It is hard, however, to imagine a macroscopic 
controller whose energy consumption does not exceed  
the energy content of the microscopic target system.
This suggests that
discussing potential thermodynamic limitations 
requires {\it microscopic} models of control. 

For both tasks, computation and control, we are criticizing basically the same issue:
as long as the device controlling or triggering the operations (regardless of whether we call it `clock' or `controller') is not included in our microscopic description, we are skeptical about the claim that the operation could `in principle' be implemented in a thermodynamic cycle
without any energy cost. 

These remarks raise the following two questions:
(1) What are appropriate models for discussing
resource requirements of computation and control?
Given such a model, we need to ask (2) how to define
{\it resource requirements} within the model.

To discuss the cost of `changing Hamiltonians' we first recall that 
changing `effective Hamiltonians' is what is actually done:
Let the target system, for instance, be a single particle. Changing control fields
actually means to change the quantum state of the physical systems surrounding the particle.
In a certain mean-field limit, this state change amounts to the change of a classical field.
Thus, the particles interact according to a {\it fixed} Hamiltonian. Taking this perspective seriously, we are looking for a model where control
operations are implemented by a fixed interaction Hamiltonian if the states of the surrounding
quantum systems are tuned in an appropriate way. Ref.~\cite{Deffner2013}
also studies thermodynamic laws in a scenario where system, controller, and 
baths are coupled by a fixed time-independent Hamiltonian, while \cite{Horowitz2014} also considers autonomuous dynamics of {\it open} systems. Although the goal of the present paper is also to 
study thermodynamics in a scenario with autonomuous time evolution, we consider a model that is 
nevertheless general enough to enable controlling controllers by `meta'-controllers and so on.
This, in turn, requires to couple the  target system 
considered in  the first place to 
an infinite system that is not just a `heat bath' as it is often assumed but something that can be controlled and, further, act as a controller at the same time.

\subsection{Spin lattice Hamiltonians as autonomous models of computation}

As models for reversible computing, Hamiltonians on spin lattices have been constructed that
are able to perform computation \cite{Margolus:90} by their autonomous evolution. 
This addresses the above criticism in the sense that these models do not require any external clocking. 
Instead, synchronization is achieved by the fixed and spatially homogeneous interaction Hamiltonian itself. Refs.~\cite{Ergodic,ErgodicQutrits} go one step further and describe Hamiltonians
on spin lattices 
for which the result of the computation need not be read out within a certain time interval
because the {\it time average state} encodes the result. This solves the more subtle problem 
that otherwise the readout required an external clock. 

There are several properties that make spin lattices attractive as physical toy models of
the world (and not only as model for a computing device): the discrete lattice symmetry represents spatial homogeneity of the physical laws and the constant Hamiltonian the homogeneity in time. By looking at lattices as discrete approximations of
a field theoretical description of the physical world, even the presence and absence of matter can be seen as just being different states of the lattice. 
Accordingly, one can argue that spin lattices allow for a quite principled way of
 studying thermodynamics of computation and control because they model not only the computing device itself but also its interaction with the environment: to this end, just consider some region in the lattice as the {\it computing device} and the complement of that region as the {\it environment}.

\subsection{Why we propose to add physical universality}

For the purpose of developing our `toy thermodynamics of computing and control' we propose to
consider spin lattices or cellular automata (as their discrete analog) that satisfy the additional condition of {\it physical universality}
introduced in \cite{PhysUniversal}.  
This property will be explained 
and motivated on an informal level in the following section. Roughly speaking, physical universality means that 
the autonomous time evolution of the system is able to implement
any mathematically possible process 
on an arbitrarily large finite region after the complement 
of the region is prepared to an appropriate initial state. 
In the case of quantum systems, we mean by `mathematically possible'
the set of completely positive trace preserving maps. In the classical case, 
we refer to the set of stochastic maps. 
Given that one believes in the hypothesis that real physical systems admit in principle the implementation of any mathematically possible process\footnote{For critical remarks on  this postulate see \cite{Omnes}, Chapter~7: here doubts are raised that every self-adjoint operator in a multi-particle system can be measured in practice. 
However, there exists always a unitary
transformation that reduces the
observable to an observable that is diagonal in the tensor product basis, i.e., measurements of every single particle. Given that one believes that these individual measurements are always possible even for multi-partite systems, the doubts thus question the implementation of arbitrary unitaries. 
Further, Ref.~\cite{Adams2017} discusses the concept of physical universality for an understanding of life and also proposes to weaken physical universality -- just to mention a second critical point of view.}, 
it is natural to demand that the interaction at hand itself is able to
implement the transformation. Otherwise, the interaction 
does not fully describe the interface between system and its environment.
For the purpose of our thermodynamic considerations, however, we 
want to study systems whose interface is completely described by the interaction under consideration rather than relying on control operations
that come as additional, external, ingredients.  

The paper is structured as follows. Section~\ref{sec:universal}
briefly motivates the notion of physical universality introduced in \cite{PhysUniversal} for both Hamiltonians and cellular automata\footnote{Note that this paper contains several ideas that already appear in the preprint \cite{PhysUniversal}, but often less explicit than here. Since \cite{PhysUniversal} will not be published 
because its main purpose had been to state a question that has been solved in the meantime, we do not care about this overlap. 
}, although we focus on the latter for sake of simplicity.
Section~\ref{sec:setting} introduces the condition of physical universality formally and describes 
and discusses the notion of resource requirements introduced in \cite{PhysUniversal}, which is also the basis of this paper. Further, we raise
the question of  whether
the resource requirements of repeating 
a certain operation can grow sublinear in the number of repetitions
(which we argue to be necessary to justify the term 'thermodynamic cycle').
Section~\ref{sec:compUniversal} explains why CAs that are 
not physically universal may admit
thermodynamic cycles in our sense.
This is because they admit
 initializations of
a finite region that ensure the implementation of endless 
repetitions of the same control operation.
Section~\ref{sec:physUniversal} explains why this simple construction  is impossible
in {\it physically universal} CAs and shows that Schaeffer's CA does not admit sublinear growth. 
Whether no
physically universal CA admits sublinear growth has to be left to the future.

\section{Physical universality: informal description and possible consequences   \label{sec:universal}}

\subsection{Physically universal systems as consistent models of control}

Ref.~\cite{PhysUniversal} introduces the notion of physical universality for three types of
systems: 

\vspace{0.3cm}
\begin{tabular}{cl}
(1) & translationally invariant finite-range interaction Hamiltonians on infinite spin\\
 &   lattices,\\
(2) &  quantum cellular automata, and \\
(3) & classical cellular automata.
\end{tabular}
 
\vspace{0.3cm} 
\noindent
 While (1) is the model that is closest to physics, (2) and (3) describe increasing abstractions that are useful for our purpose.
Essentially, (2) is just the discrete time version of (1). We will restrict the attention
to (3) because it turns out that the problem is already difficult enough for this case.

On an abstract level, the definition of physical universality coincides for all three cases:
a system is called physically universal if every desired transformation on any desired target region (of  
arbitrary but finite size)
can
be implemented by first initializing the (infinite) complement of that region to an appropriate state
and then letting the system evolve according to its autonomous dynamics for a certain `waiting time' $t$.
For the cases (2) and (3), $t$ is a positive integer while it is a positive real number for the case (1). 
Since cases (1) and (2) refer to {\it quantum} systems the set of possible transformations
(completely positive trace preserving maps) is uncountably infinite, we should only demand that 
one can get {\it arbitrarily close} to the desired transformation via appropriate initializations and
waiting times instead of being able to implement the desired transformation exactly.

\paragraph{Shifting the boundary between
target and controller}

Physically universal systems are intriguing because they provide a model class where 
every physical degree of freedom is indirectly 
accessible by operating on the remaining degrees of freedom in the `world' and then letting the joint system evolve. In other words, the complement of the target region acts as the controller
of the target region so that any part of the world can become the controller or the  
system to be controlled. This is in contrast to some physical models of computation, e.g.,
\cite{Ergodic}, for which data and program registers are represented by different types of physical degrees of freedom. These systems are able to perform any desired transformation on the {\it data register} by appropriate initialization of the {\it program register}. The question of how to
act on the program register cannot be addressed within the model.
In physically universal systems, on the other hand, the preparation of
{\it any} region can be achieved by operating on its complement. 
This reduces the question of how to act on some target region to the  question of
how to act on some `controller' region around it. In turn, this controller region can be
prepared by acting on some `meta-controller' region around it.
Although this does not solve the problem it shows
at least that the boundary between controller and target region can be arbitrarily shifted.

\paragraph{Analogy to the quantum measurement problem}

This is similar to the
quantum measurement problem where the boundary between the measurement apparatus and the quantum system to be measured (the famous `Heisenberg cut') can be arbitrarily shifted as long as the quantum description is considered appropriate:
the transition 
from a pure superposition to
the corresponding mixture can be explained by entanglement between the target system
and its measurement aparatus \cite{Guilini96} (for simplicity, one may define `measurement apparatus' as all parts of the environment that carry information about the result). The resulting joint superposition of measurement apparatus 
and target system can be transferred to a mixture by entanglement with a `meta' measurement apparatus and so on. 

\subsection{Potential thermodynamic implications}

Physical universality can have important thermodynamic consequences because
it excludes the ability to completely protect information. 
Physically universality means that any system can be controlled by its surrounding. Therefore,  
the unknown state of the surrounding will eventually cause the state of the system to change.
In contrast, in systems such as \cite{Ergodic} the state of the program register never changes during the autonomous because of the strict separation between data and program registers.
Here, we don't want to accept the latter class of models as physical models of computation because 
in the real world also program registers are physical systems that can be somehow accessed by actions on their environment. In other words, the information of the `program' register is only safe because the model fails to describe how to act on that part of the system
using the given interactions (these actions are external to the theory).

\paragraph{Trade-off between stability and
controllability} 

Physical universality thus gives rise to a thermodynamics 
in which the {\it inability to protect} information is a result of the 
{\it ability to control} every degree of freedom. 
On the one hand, the target system needs to interact with its environment otherwise we were not able to control it. On the other hand, this interaction makes entropy leaking from the surrounding into the target system. Ref.~\cite{PhysUniversal} defines the model class of physically universal systems for the purpose of studying this conflict on an abstract level. Here, we restrict the attention
to discrete time dynamics on {\it classical cellular automata}. In the long run, one should certainly address our thermodynamic questions using {\it continuous time dynamics on quantum systems}. As a first approach, however, it is convenient to simplify the problem by restricting oneself to classical CAs.
Another reason for considering classical CAs is also to make this problem more accessible to the computer science community.\footnote{Note, further, that 
already von Neumann's self-reproducing automata
\cite{vonNeumann1966} follows the principle
to study physical or biological
universality properties using CAs.}
After all, it is one of the lessons learned from quantum information theory \cite{NC} that translating physics into computer scientific language can provide a new perspective and new paradigms. Indeed, the past two decades have shown that understanding thermodynamics via computer scientific models is also promising.\footnote{For instance, the principle of cooling devices \cite{Fernandez,dynamical_cooling} and heat engines
\cite{HeatEngines} can be illustrated using an $n$-bit register represented by $n$ two-level systems or other simple discrete systems. For this model class, the relation between physics and information is most obvious.} On the microscopic level one can hardly tell apart computing devices from thermodynamic machines in the conventional sense.\footnote{See also the adaptive heat engine in Ref.~\cite{Allahverdyan16}.}   
As part of this oversimplification, we will define the thermodynamic  cost of an operation
simply by the size of the region in the surrounding of the target system that needs to be initialized. This will be partly justified in Section~\ref{subsec:cost}.

\section{The formal setting \label{sec:setting}}

\subsection{Notation and terminology\label{subsec:notation}}

For the basic notation we mainly follow \cite{Schaeffer2015}. The cells of our CA in $d$ dimensions are located at lattice points
in $\Omega:=\Z^d$. The {\it state} of each cell is 
given by an element of the alphabet $\Sigma$. For any subset $X \subset \Omega$,
a {\it configuration} $\gamma_X$ of $X$ is
a map $X \to \Sigma$. Let $\Sigma^X$ denote the
set of all configurations of $X$. 
The dynamics of the CA is given by a map $\alpha:\Sigma^\Omega \to
\Sigma^\Omega$ that is local (i.e. 
the state of each cell is only influenced by the state of cells in a fixed neighborhood) 
and spatially homogeneous (i.e., it commutes with all lattice translations). 

Later, we will often consider a class of CAs 
in dimension $d=2$ where the state of a cell one time step later only depends on the state of the cell
itself and its $8$ surrounding neighbors,
the so-called Moore neighborhood, and refer to this class as `Moore CAs'.

If
$\gamma' :=\alpha (\gamma)$
for any $\gamma \in \Sigma^\Omega$,
we also write
$\gamma \to \gamma'$ to indicate that
the configuration $\gamma$ evolves to $\gamma'$ in one time step and $\gamma \stackrel{n}{\to} \gamma' = \alpha^n (\gamma)$ means that $\gamma$ evolves to $\gamma'$ in $n$ time steps.

\begin{Definition}[implementing a function]
Let $X,Y \subset \Omega$ be finite sets
and $f: \Sigma^X \to \Sigma^Y$ be an arbitrary function.
Then we say a configuration
$\phi\in \Sigma^{\bar{X}}$ implements $f$
in time $t$ if for every $x\in \Sigma^X$
\[
\phi \oplus x \stackrel{t}{\mapsto} \psi_x \oplus f(x), 
\]
holds
for some $\psi_x \in \Sigma^{\bar{Y}}$.  Here, the sign $\oplus$ denotes merging configurations
of disjoint regions to a configuration of the union.  
\end{Definition}

For physical universality, we follow
Schaeffer's modified definition \cite{Schaeffer2015}, which is equivalent to
our original one, and also his definition of
{\it efficiently} physically universal:

\begin{Definition}[physical universality]
We say a cellular automaton is physically universal if for all finite regions $X,Y$ and all transformations $f:\Sigma^X \to \Sigma^Y$, there exists a configuration 
$\phi$ of the complement of $X$ and a natural number $t \in \N$ such that $\phi$
implements $f$ in time $t$.

We say the CA is efficiently physically universal if the implementation runs in time $t_0$, where
$t_0$ is polynomial in

\noindent
$\bullet$ the diameter of $X$ (i.e., the width of the smallest hypercube containing the set) and
diameter of $Y$,

\noindent
$\bullet$ the distance between $X$ and $Y$, and

\noindent
$\bullet$ the computational complexity of $f$ under some appropriate model of computation (e.g.,
the number of logical gates in a circuit for $f$).
\end{Definition}

For simplicity, we will often consider only the case $Y=X$. Since every signal in our CA propagates only one cite per time step, at most 
 a margin of thickness $t$ around
$X$ matters for what happens after $t$ time steps. 
Depending on the dynamical law and the desired operation on the target region, the relevant part of the state can be significantly less.
To explore the resource requirements of an 'implementation'
we phrase the notion of an implementation formally in a way that is explicit about which parts of the surrounding cells
matter to achieve the desired operation:

\begin{Definition}[device for implementing $f$]\label{def:implement}
A device for implementing $f: \Sigma^X \rightarrow \Sigma^Y$ 
is a triple
$(Z,\phi_Z,t)$ 
such that
$\phi_Z \oplus \phi'$ implements $f$ in
$t$ time steps for all $\phi' \in \Sigma^{\bar{Z} \cap \bar{X}}$. 
Here, $X$ and $Y$ are called the `source region' and `target region', respectively, and
$Z \subset \bar{X} $ is called the `relevant region', $\phi_Z  \in \Sigma^{Z}$ the state of this region, and
$t \in \N$ the `implementation time'.
Then, the `size' of the device is the size of $W:=Z \cup X\cup Y$.
The `range' of the device is the
side length of the smallest $d$-dimensional hypercube containing $W $.
\end{Definition}

Note that the relevant region may overlap with the target region while it needs to be disjoint of the source region. 
Further, note that the definition of a device does not imply
that the relevant region has been chosen in a minimal way. 
Accordingly, future theorems on the resource requirements of implementations may read `the relevant region consists of
at least \dots  cells.'
The range can be seen as the size of the 
apparatus.  
Assume, for instance, that $W$ consists 
of a small number $n$ of single cells spread over
a hypercube of side length $k\gg n$. Then we would still call this
a `large' apparatus even if $n$ is small. 

So far, we have only considered the ability to implement one specific transformation once. We also want to be able to study processes where one desired  operation is performed
after time $t_1$, a second one after time $t_2+t_1$, and so on. 
Assume, for instance, that we want to achieve that 
the information content of a certain cell $c_1\in \Omega$ 
is shifted to cell $c_2$ after some time $t_1$ and then shifted to cell $c_3$ at some later time $t_2+t_1$. Then the entire process
consisting should be performed by one initialization
rather than demanding re-preparing the system after each transformation. 
To this end, 
we define devices for implementing concatenations of transformations
as generalization of Definition~\ref{def:implement}:

\begin{Definition}[device for implementing a sequence of transformations]\label{def:implSeq}
Let $X_1,\dots,X_{n+1}$ be finite regions
and $f_1,\dots,f_n$ be functions with $f_j:\Sigma^{X_j} \to \Sigma^{X_{j+1}}$ for $j=1,\ldots,n$. 
In other words, the target region of $f_j$ is the source region of $f_{j+1}$.   
A device for implementing the sequence $f_1,f_2,\dots,f_n$ is an
$n+2$-tuple 
$(Z,\phi_Z,t_1,\dots,t_n)$ with $t_j>0$, 
where  $Z\subset \bar{X}_1 $ is called the `relevant region'
and $\phi_Z \in \Sigma^Z$ is a configuration such that
$\phi_Z \oplus \phi'$ implements $f_j \circ f_{j-1} \circ \cdots \circ f_1$ in
$\sum_{i=1}^j t_i$ time steps for all $\phi' \in \Sigma^{\bar{Z} \cap \bar{X}_1}$. 
The size of the device is the size of
$W:=Z \cup  (\cup_{j=1}^{n+1} X_j)$ and its range is the
side length of the smallest $d$-dimensional hypercube containing $W$.
\end{Definition}

The idea of  Definition~\ref{def:implSeq} 
is that the CA implements the transformation $f_j$ 
within $t_j$ time steps, but this interpretation 
can be misleading because the Definition only specifies that the initial state $x$ is transformed into the final state 
\[
f_n(f_{n-1}(\cdots f_1(x)\cdots))
\]
 if the CA is not disturbed during the entire process.  This does not require, for instance, that an external intervention that changes the state of the region  $X_1$
from $f_1(x)$ to some $y$ between step $t_1$ and $t_1+1$
yields the final state $f_n(f_{n-1}(\cdots f_2(y)\cdots))$.\footnote{Rephrased in causal language \cite{Pearl2000}, if we denote the
state of $X_j$ at time $\sum_{i=1}^j t_i$ by $x_j$, then  
the equation 
\begin{equation}\label{eq:sem}
x_j = f_j(x_{j-1}),
\end{equation}
is  not a `structural equation', since
the latter describes, by definition,  also the
impact of interventions on the input variable
on the right hand side.} 

A priori it is not obvious that physical
universality entails the ability of implementing sequences with $n>1$. Th following result shows that this is the case:

\begin{Theorem}[ability to implement sequences]
In every physically universal CA there is a device
for any sequence of transformations. 
\end{Theorem}

\begin{proof}
We provide a proof by induction on $n$.  The base case $n=1$ follows from physical universality.
For the induction hypothesis assume that sequences of $n$ arbitrary functions can be implemented. 

For the induction step, let $f_1,\ldots,f_n,f_{n+1}$ be a sequence of $n+1$ arbitrary transformations with 
\[
f_j : \Sigma^{X_j} \rightarrow \Sigma^{X_{j+1}}
\]
for $j=1,\ldots,n+1$.

By physical universality there exists a device $(Z_{n+1},\phi_{Z_{n+1}},t_{n+1})$ with $Z_{n+1} \subset \bar{X}_{n+1}$ that implements the last 
function $f_{n+1}$ of the above sequence.  Using this device we define the following augmented version $\hat{f}_n$ of the second last function $f_n$ of the above sequence by setting
\[
\hat{f}_n : \left\{
\begin{array}{ccccc}
\Sigma^{X_n} & \rightarrow & \Sigma^{X_{n+1}} & \cup   & \Sigma^{Z_{n+1}} \\ 
          x          & \mapsto     & f_n(y)                     & \oplus & \phi_{Z_{n+1}}
\end{array}
\right.
\]
for all $y\in\Sigma^{X_n}$. In words,
the output of the augmented function $\hat{f}_n$ consists of the output of original function $f_n$ on the region $X_n$ and the constant output $\phi_{Z_{n+1}}$ on the region $Z_{n+1}$.  

By induction hypothesis there exists a device $(Z,\phi,t_1,\ldots,t_{n-1},t_n)$ that implements the sequence $f_1,\ldots,f_{n-1},  \hat{f}_n$.  The special form of the output of the augmented function $\hat{f}_n$ ensures that the device
$(Z,\phi,t_1,\ldots,t_{n-1},t_n,t_{n+1})$ also implements the sequence $f_1,\ldots,f_{n-1},f_n,f_{n+1}$. This is because after $t_1+\ldots +t_n$ times steps the output is
\[
\hat{f}_n(y) = f_n(y) \oplus \phi_{Z_{n+1}} \in \Sigma_{X_{n+1}} \cup \Sigma_{Z_{n+1}} \quad \mbox{where} \quad y = f_{n-1}(\ldots (f_1(x)) \ldots )
\]
so that after $t_{n+1}$ additional time steps the final output is 
\[
f_{n+1}(f_n(y))\in\Sigma^{X_{n+2}}
\]
as desired.
\end{proof}

To mention a simple example of the kind of sequences we are interested in, consider a CA with binary alphabet $\Sigma=\{0,1\}$.
Assume 
the task is to implement a NOT gate on the same bit $n$ times on some target bit. 
Then the desired functions read 
$f_j={\rm NOT}$
and 
the numbers $t_j$
specify the time instants for which the autonomous dynamics has implemented another NOT gate on our target bit, given that some region $Z$ has been initialized to the state $\phi_Z$.  

\subsection{Formalizing `thermodynamic cost' of operations \label{subsec:cost}}

Here we will consider the size of the relevant region as the thermodynamic
cost of an implementation. This first approximation is justified by the following idea: a priori, the state
of each cell is unknown, i.e., we assume uniform distribution over $\Sigma$. According to
Landauer's principle it then requires the energy $k T \ln |\Sigma|$ to initialize one cell
to the desired state. This way, the thermodynamic cost of the initialization process
is simply proportional to the number of cells to be initialized.
This view will be further discussed at the end of this subsection. 

Note that the size of the relevant region can only grow with $O(t^d)$ if $t$ is the running time for an implementation since a signal can only proceed a constant number of cells per time step. Therefore, the thermodynamic 
cost scales only polynomial in the computational complexity if a CA is efficiently physically universal. This statement, however, is too weak for our purpose.
To phrase the main questions of this paper (which look for stronger statements) 
we need the following terminology: 

\begin{Definition}[zero cost per operation]
Given a function $f: \Sigma^X \to \Sigma^X$, 
a physically universal CA is said to admit
the implementation of $f$ at zero cost per operation, if 
there are devices 
$(Z_n,\phi_{Z_n},t_1,\dots,t_n)$ for every $n\in \N$, each implementing $f$, such that
\[
\lim_{n\to \infty} \frac{|Z_n|}{n} =0.
\]
\end{Definition}
Note that this definition does not require that the
implementation of $f$ stops after the time $t_n$. 
Likewise, we define:
\begin{Definition}[zero cost of information storage per time]
For some region $X$,
a physically universal CA is said to
admit zero cost of information storage per time on $X$ if 
there are devices $(Z_n,\phi_{Z_n},t_n)$ 
for every $n\in \N$ with $t_n\to \infty$ 
that implement the identity on $X$ after the time $t_n$
such that
\[
\lim_{n\to \infty} \frac{|Z_n|}{t_n} =0.
\]
\end{Definition}

\noindent
We are now able to phrase our main questions:

\begin{itemize}
\item {\bf Question 1:} 
Does there exist a physically universal CA
that admits zero cost per operation for
any / for all functions $f$?

\item {\bf Question 2:}
Does there exist a physically universal CA
that admits zero cost for information storage per time for any / for all finite regions $X$?
\end{itemize}

\noindent
If we recall that the state of the CA may also encode the presence or absence of matter, our definition of implementation cost also includes the aspect of hardware deterioration. Assume
one has built some microscopic control device that degrades
after performing an operation some large number $n_0$ of times, a device for implementing the operation $n>n_0$ times includes
a `meta' device repairing the original one. 
\footnote{Thermodynamic considerations
that account also for reproduction processes 
are certainly related to thermodynamics of life \cite{Kempes2017}.}

On the one hand, we will show that the answers are negative for Schaeffer's CA  \cite{Schaeffer2015} to both questions above. On the other hand, we will show that there exist physically non-universal CAs for which both answers are positive. We leave it as an open question whether 
physically universality precludes the ability to achieve zero cost.
However, we give some intuitive arguments that suggest that physical universality makes it at least more difficult to achieve zero implementation cost per operation 
or zero cost for information storage per time. 

\paragraph{Discussion of the above formalization of thermodynamic cost}
It is certainly an oversimplification to identify the size of the region that needs to be initialized with the thermodynamic cost of an implementation. Consider, for instance, 
a physical many particle system where each cell is a physical system that is weakly interacting
with its neighbors. This ensures that the total energy of the composed system
is approximately given by the sum of the energy of the individual systems.  
Assume, furthermore, that
the state $0\in \Sigma$ corresponds to the ground state, that is, the state of 
lowest energy. In the limit of low temperatures, this state has probability close to $1$, which implies that initializing the lattice to the all-zero state 
does not require significant free energy resources. In this case, however, it requires
significant free energy resources to set a cell to {\it any state  other than} $0$
and the resource requirement then depends on the number of cells that need to be in a non-zero state (which may correspond to the number of particles in physics).

On the other hand, identifying the number of cells to be initialized with the thermodynamic cost, can also be justified from the following point of view: assume we are not interested in the amount of free energy that is required for one specific transformation. Instead, we only ask whether
the amount increases sublinearly or not. Assuming, in the above physical picture, 
non-zero temperature (although it may still be low, which favors the state $0$), initializing
$n$ states to $0$ {\it with certainty} yet requires an amount of free energy {\it of the order} $n$. This way, the asymptotic behavior of resource requirements is unaffected  by the details of the physical hardware assumptions. 

\section{Cost of operations in Turing complete CAs\label{sec:compUniversal}}

As a simple toy example, we consider the control task of repeatedly turning on and off a target bit without ever stopping. 
Intuitively, this process already reminds us of a program with an infinite loop:

\begin{Example}[infinite bit switching]
\begin{algorithmic}
 \State
 \State $a:=0$ 
 \While{$(\,1\,)$ }
   \State $a:=1 \oplus a$ \quad // bit XOR
 \EndWhile
\end{algorithmic}
\end{Example}
\noindent
Every Turing-complete CA is capable of implementing the above program. We now explain briefly the notion of Turing-complete CAs.   
A CA is called Turing-complete if there exists a finite configuration
that allows the CA to simulate any universal Turing machine, where the concepts of
`finite configuration' and `halting' are defined as follows.   
`Finite configuration' means that only finitely many cells are in a non-zero state, where a single element of the alphabet $\Sigma$ is chosen to be zero, denoted by $0$.
`Halting' is defined as the event of a single previously selected cell becoming non-zero.

It is important to observe that {\it finite configuration} does not imply {\it finite resources} in our sense.
`Finite configuration' means that all but a finite number of cells are in the {\it zero} state, whereas 
`finite resources' means that all but a finite number of cells are in an {\it unknown} state. 

Consider the following situtation: the simulation of a universal Turing machine by a CA could {\it require} that all but a finite number of cells be zero because otherwise the non-zero cells would eventually perturb the simulation.  This would mean {\it infinite} resources in our sense. 
However, as long as we do not demand physical universality, we can easily modify Turing complete
CAs such that they are able to implement an infinite loop with finite resources,
as will be discussed in the following two subsections.  

\subsection{Conway's Game of Life}

We first consider the implementation of our target operation  `infinite bit switch' in a well-known cellular automaton,
namely Conway's Game of Life. It is a CA in two dimensions, each cell being `alive' or 'dead', i.e.,
formally each cell is just one bit. The rules are \cite{GameOfLifeWiki}: 

\noindent
(1) Any live cell with fewer than two live neighbours dies, as if caused by under-population.

\noindent
(2) Any live cell with two or three live neighbours lives on to the next generation.

\noindent
(3) Any live cell with more than three live neighbours dies, as if by over-population.

\noindent
(4) Any dead cell with exactly three live neighbours becomes a live cell, as if by reproduction.

To implement the bit flip, as desired, we find
simple oscillating patterns in \cite{GameOfLifeWiki}:
The `Blinker' has period 2, as shown in Figure~\ref{fig:blinker}.
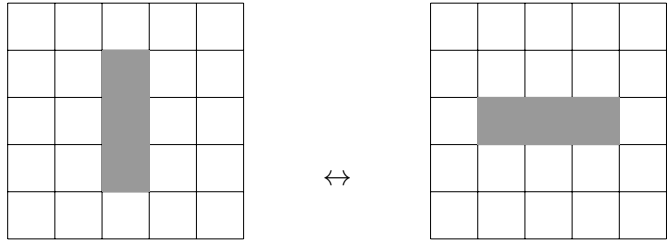
\begin{figure}
\centerline{
\resizebox{0.2\textwidth}{!}{
\begin{tikzpicture}
\draw[step=1cm,black,very thin] (-2,-2) grid (3,3);
\fill[black!40!white] (0,0) rectangle (1,1);
\fill[black!40!white] (0,-1) rectangle (1,1);
\fill[black!40!white] (0,2) rectangle (1,1);
\end{tikzpicture}
}
\hspace{0.5cm}
\begin{tabular}{c}
$\leftrightarrow$\\
${}$\\
${}$\\
\\
\end{tabular}
\hspace{0.5cm}
\resizebox{0.2\textwidth}{!}{
\begin{tikzpicture}
\draw[step=1cm,black,very thin] (-2,-2) grid (3,3);
\fill[black!40!white] (0,0) rectangle (1,1);
\fill[black!40!white] (-1,0) rectangle (1,1);
\fill[black!40!white] (2,0) rectangle (1,1);
\end{tikzpicture}
}
}
\caption{\label{fig:blinker}
A simple configuration in Conway's Game of Life that yields a dynamical behavior with period $2$. The system changes between the two configurations on the left and the right hand side, respectively. `Alive' and `dead' cells are indicated by gray and white, respectively.} 
\end{figure}
We now focus on the space requirements
of this 2-cycle and recall that
space requirements in our sense refer to the amount of space that needs to be initialized to a specific value. 
For the Blinker to work, it is essential that there are no `particles' in the direct neighborhood that disturb the patterns. 
Whenever there is a region 
outside which the state is not known at all, this complementary region contains with some probability a pattern that moves towards the blinker and disturbs its cycle.  It is therefore possible, that, without having some control about the entire space, we cannot guarantee that the blinker works forever.

%
%
\subsection{Modified Game of Life with impenetrable walls}

There is, however, a simple modification of the Game of Life for which we can ensure that the blinker works forever although we only control the state of a finite region. To this end, we augment each cell by an additional third state `brick' $\blacksquare$, indicated  by black color,
that blocks the diffusion from the surrounding.  
The transition rule of the new CA now consist of
the following rules:

\noindent
(0) a cell being in the state $\blacksquare$ remains there forever. 
(1)-(4) as before, with the convention that the brick $\blacksquare$ counts
as $\square$ for its neighbors. 

The idea of bricks is that they can form a `wall'
around our blinker that protects it from 
the influence of its surrounding (which can be in an unknown state).
In physical terms, the wall protects the blinker from the heat of the environment, as shown in Figure~\ref{fig:withwall}.

\begin{figure}
\centerline{
\resizebox{0.3\textwidth}{!}{
\begin{tikzpicture}
\draw[step=1cm,black,very thin] (-3,-3) grid (6,6);
\fill[black!40!white] (1,1) rectangle (2,2);
\fill[black!40!white] (0,1) rectangle (1,2);
\fill[black!40!white] (2,1) rectangle (3,2);
\fill[black] (-2,-2) rectangle (-1,-1);
\fill[black] (-1,-2) rectangle (0,-1);
\fill[black] (0,-2) rectangle (1,-1);
\fill[black] (1,-2) rectangle (2,-1);
\fill[black] (2,-2) rectangle (3,-1);
\fill[black] (3,-2) rectangle (4,-1);
\fill[black] (4,-2) rectangle (5,-1);
\fill[black] (4,-1) rectangle (5,0);
\fill[black] (4,0) rectangle (5,1);
\fill[black] (4,1) rectangle (5,2);
\fill[black] (4,2) rectangle (5,3);
\fill[black] (4,3) rectangle (5,4);
\fill[black] (4,4) rectangle (5,5);
\fill[black] (3,4) rectangle (4,5);
\fill[black] (2,4) rectangle (3,5);
\fill[black] (1,4) rectangle (2,5);
\fill[black] (0,4) rectangle (1,5);
\fill[black] (-1,4) rectangle (0,5);
\fill[black] (-2,4) rectangle (-1,5);
\fill[black] (-2,3) rectangle (-1,4);
\fill[black] (-2,2) rectangle (-1,3);
\fill[black] (-2,1) rectangle (-1,2);
\fill[black] (-2,0) rectangle (-1,1);
\fill[black] (-2,-1) rectangle (-1,0);
\end{tikzpicture}
}
}
\caption{\label{fig:withwall} The blinker surrounded by a wall of `bricks', which protect it from uncontrolled perturbations from its environment.}
\end{figure}
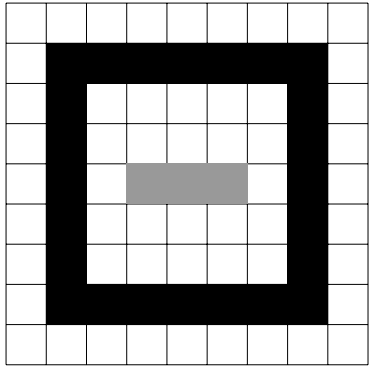

\subsection{Reversible CA: Margolus' billard ball model \label{subsec:mbbm}}

To get one step closer to physics and
account for the bijectivity of microscopic dynamics
in the physical world, we now consider reversible CAs, i.e.,
CAs in which every state has a unique predecessor, which is not the case for Game of Life. 
We now show that even reversible CAs exist that admit perfect protection of an implementation of
an infinite loop, which results in zero cost per operation.

{\it Margolus' billard ball model CA} \cite{Margolus84} is a CA in $2$ dimensions whose update rules are defined
on Margolus neighborhoods, i.e., there are two partitions of the grid
into blocks of $2\times 2$ cells
describing the updates at even and odd time instants:
At even time instances, the update is done
on the blocks $\{(2i,2j),(2i,2j+1),(2i+1,2j), (2i+1,2j+1)\}$, at odd times it is done
on the blocks
$\{(2i-1,2j-1),(2i-1,2j),(2i,2j-1),(2i,2j)\}$, as visualized by the black and the red grid
in Figure~\ref{fig:bbmca}, right.
For each such block, the update rules are shown in Figure~\ref{fig:bbmca}, left. 
To interpret such a CA with Margolus neighborhood as a so-called Moore CA  where the update rules do not change between even and odd time steps (see Subsection~\ref{subsec:notation}), we consider two time steps in the Margolus CA as one time step of a  Moore CA.
To ensure that the update of a cell of the Moore CA only depends on its surrounding neighbors
(which is convenient for some purposes)
one may consider each $2\times 2$ block of the Margolus CA as one cell of the Moore CA. 
 
\begin{figure}
\centerline{
\begin{tabular}{cc}
(1) & 
\resizebox{0.2\textwidth}{!}{
\begin{tikzpicture}
\draw[step=1cm,black,very thin] (0,0) rectangle (2,2);
\draw[dotted] (0,1) -- (2,1);
\draw[dotted] (1,0) -- (1,2);
\draw[very thick,->] (2.5,1) -- (3.5,1);
\draw[step=1cm,black,very thin] (4,0) rectangle (6,2);
\draw[dotted] (4,1) -- (6,1);
\draw[dotted] (5,0) -- (5,2);
\end{tikzpicture}
}
\\
(2) & \resizebox{0.2\textwidth}{!}{
\begin{tikzpicture}
\draw[step=1cm,black,very thin] (0,0) rectangle (2,2);
\draw[dotted] (0,1) -- (2,1);
\draw[dotted] (1,0) -- (1,2);
\fill[black!40!white] (0,1) rectangle (1,2);
\draw[very thick,->] (2.5,1) -- (3.5,1);
\draw[step=1cm,black,very thin] (4,0) rectangle (6,2);
\draw[dotted] (4,1) -- (6,1);
\draw[dotted] (5,0) -- (5,2);
\fill[black!40!white] (5,0) rectangle (6,1);
\end{tikzpicture}
}
\\
(3) & 
\resizebox{0.2\textwidth}{!}{
\begin{tikzpicture}
\draw[step=1cm,black,very thin] (0,0) rectangle (2,2);
\draw[dotted] (0,1) -- (2,1);
\draw[dotted] (1,0) -- (1,2);
\fill[black!40!white] (0,1) rectangle (1,2);
\fill[black!40!white] (1,0) rectangle (2,1);
\draw[very thick,->] (2.5,1) -- (3.5,1);
\draw[step=1cm,black,very thin] (4,0) rectangle (6,2);
\draw[dotted] (4,1) -- (6,1);
\draw[dotted] (5,0) -- (5,2);
\fill[black!40!white] (4,0) rectangle (5,1);
\fill[black!40!white] (5,1) rectangle (6,2);
\end{tikzpicture}
}
\\
(4) & 
\resizebox{0.2\textwidth}{!}{
\begin{tikzpicture}
\draw[step=1cm,black,very thin] (0,0) rectangle (2,2);
\draw[dotted] (0,1) -- (2,1);
\draw[dotted] (1,0) -- (1,2);
\fill[black!40!white] (0,1) rectangle (1,2);
\fill[black!40!white] (0,0) rectangle (1,1);
\draw[very thick,->] (2.5,1) -- (3.5,1);
\draw[step=1cm,black,very thin] (4,0) rectangle (6,2);
\draw[dotted] (4,1) -- (6,1);
\draw[dotted] (5,0) -- (5,2);
\fill[black!40!white] (4,0) rectangle (5,1);
\fill[black!40!white] (4,1) rectangle (5,2);
\end{tikzpicture}
}
\\
(5) & 
\resizebox{0.2\textwidth}{!}{
\begin{tikzpicture}
\draw[step=1cm,black,very thin] (0,0) rectangle (2,2);
\draw[dotted] (0,1) -- (2,1);
\draw[dotted] (1,0) -- (1,2);
\fill[black!40!white] (1,1) rectangle (2,2);
\fill[black!40!white] (1,0) rectangle (2,1);
\fill[black!40!white] (0,0) rectangle (1,1);
\draw[very thick,->] (2.5,1) -- (3.5,1);
\draw[step=1cm,black,very thin] (4,0) rectangle (6,2);
\draw[dotted] (4,1) -- (6,1);
\draw[dotted] (5,0) -- (5,2);
\fill[black!40!white] (4,0) rectangle (5,1);
\fill[black!40!white] (5,1) rectangle (6,2);
\fill[black!40!white] (5,0) rectangle (6,1);
\end{tikzpicture}
}
\\
(6) & 
\resizebox{0.2\textwidth}{!}{
\begin{tikzpicture}
\draw[step=1cm,black,very thin] (0,0) rectangle (2,2);
\draw[dotted] (0,1) -- (2,1);
\draw[dotted] (1,0) -- (1,2);
\fill[black!40!white] (1,1) rectangle (2,2);
\fill[black!40!white] (1,0) rectangle (2,1);
\fill[black!40!white] (0,0) rectangle (1,1);
\fill[black!40!white] (0,1) rectangle (1,2);
\draw[very thick,->] (2.5,1) -- (3.5,1);
\draw[step=1cm,black,very thin] (4,0) rectangle (6,2);
\draw[dotted] (4,1) -- (6,1);
\draw[dotted] (5,0) -- (5,2);
\fill[black!40!white] (4,0) rectangle (5,1);
\fill[black!40!white] (5,1) rectangle (6,2);
\fill[black!40!white] (5,0) rectangle (6,1);
\fill[black!40!white] (4,1) rectangle (5,2);
\end{tikzpicture}
}
\end{tabular}
\hspace{2cm}
\begin{tabular}{c}
\resizebox{0.3\textwidth}{!}{
\begin{tikzpicture}
\draw[step=2cm,black,thin] (0,0) grid (8,8);
\draw[red,thin] (1,0) -- (1,8);
\draw[red,thin] (3,0) -- (3,8);
\draw[red,thin] (5,0) -- (5,8);
\draw[red,thin] (7,0) -- (7,8);
\draw[red,thin] (0,1) -- (8,1);
\draw[red,thin] (0,3) -- (8,3);
\draw[red,thin] (0,5) -- (8,5);
\draw[red,thin] (0,7) -- (8,7);
\end{tikzpicture}
}
\end{tabular}
}
\caption{\label{fig:bbmca} Left: Transition rules of Margolus' billiard ball model CA. Right:
the two different partitions are indicated by the black and the red grid.}
\end{figure}
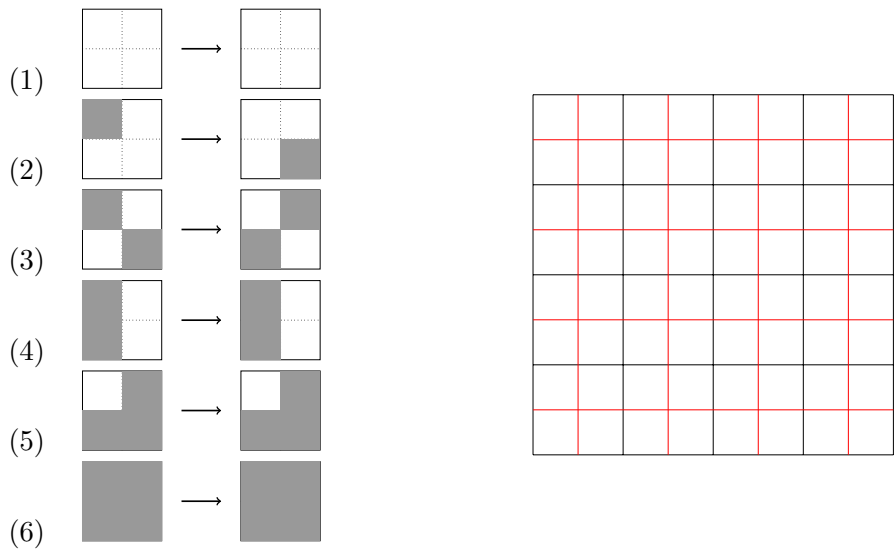

As noted in \cite{Schaeffer2015},
the billiard ball CA is not physically universal since it allows for impenetrable walls \cite{Margolus84}.
We will use such walls to implement a bit switching process that continues forever although only a finite region has been initialized. 
A simple example is shown in Figure~\ref{fig:switchingbbmca}.
In the sense of the present paper, this CA implements the NOT
operation in a thermodynamic cycle since there are no resource requirements per operation because there is no need to initialize the cells outside the wall. 

\begin{figure}
\centerline{
\resizebox{0.4\textwidth}{!}{
\begin{tikzpicture}
\draw[step=2cm,black,thin] (0,0) grid (10,10);
\draw[red,thin] (1,0) -- (1,10);
\draw[red,thin] (3,0) -- (3,10);
\draw[red,thin] (5,0) -- (5,10);
\draw[red,thin] (7,0) -- (7,10);
\draw[red,thin] (9,0) -- (9,10);
\draw[red,thin] (0,1) -- (10,1);
\draw[red,thin] (0,3) -- (10,3);
\draw[red,thin] (0,5) -- (10,5);
\draw[red,thin] (0,7) -- (10,7);
\draw[red,thin] (0,9) -- (10,9);
\fill[black!40!white] (1,2) rectangle (3,4);
\fill[black!40!white] (3,2) rectangle (5,4);
\fill[black!40!white] (5,2) rectangle (7,4);
\fill[black!40!white] (7,2) rectangle (9,4);
\fill[black!40!white] (1,4) rectangle (3,6);
\fill[black!40!white] (7,4) rectangle (9,6);
\fill[black!40!white] (1,6) rectangle (3,8);
\fill[black!40!white] (3,6) rectangle (5,8);
\fill[black!40!white] (5,6) rectangle (7,8);
\fill[black!40!white] (7,6) rectangle (9,8);
\fill[black!40!white] (4,4) rectangle (5,5);
\fill[black!40!white] (5,5) rectangle (6,6);
\fill[black!40!white] (3,4) rectangle (4,5);
\fill[black!40!white] (3,5) rectangle (4,6);
\fill[black!40!white] (6,4) rectangle (7,5);
\fill[black!40!white] (6,5) rectangle (7,6);
\end{tikzpicture}
}
\hspace{1cm}
\resizebox{0.4\textwidth}{!}{
\begin{tikzpicture}
\draw[step=2cm,black,thin] (0,0) grid (10,10);
\draw[red,thin] (1,0) -- (1,10);
\draw[red,thin] (3,0) -- (3,10);
\draw[red,thin] (5,0) -- (5,10);
\draw[red,thin] (7,0) -- (7,10);
\draw[red,thin] (9,0) -- (9,10);
\draw[red,thin] (0,1) -- (10,1);
\draw[red,thin] (0,3) -- (10,3);
\draw[red,thin] (0,5) -- (10,5);
\draw[red,thin] (0,7) -- (10,7);
\draw[red,thin] (0,9) -- (10,9);
\fill[black!40!white] (1,2) rectangle (3,4);
\fill[black!40!white] (3,2) rectangle (5,4);
\fill[black!40!white] (5,2) rectangle (7,4);
\fill[black!40!white] (7,2) rectangle (9,4);
\fill[black!40!white] (1,4) rectangle (3,6);
\fill[black!40!white] (7,4) rectangle (9,6);
\fill[black!40!white] (1,6) rectangle (3,8);
\fill[black!40!white] (3,6) rectangle (5,8);
\fill[black!40!white] (5,6) rectangle (7,8);
\fill[black!40!white] (7,6) rectangle (9,8);
\fill[black!40!white] (4,5) rectangle (5,6);
\fill[black!40!white] (5,4) rectangle (6,5);
\fill[black!40!white] (3,4) rectangle (4,5);
\fill[black!40!white] (3,5) rectangle (4,6);
\fill[black!40!white] (6,4) rectangle (7,5);
\fill[black!40!white] (6,5) rectangle (7,6);
\end{tikzpicture}
}
}
\caption{\label{fig:switchingbbmca}
Configuration of Margolus billiard ball CA
 that implements
bit switching forever: applying Rule (3) to the black partitioning takes the configuration on the left hand side to the one on the right hand side. Then,  an update according to the red partitioning leaves the state unchanged
due to Rule (5). 
Applying Rule (3) to the black partitioning takes the configuration on the right hand side back to the one on the left hand side. Again, updating according to the red partitioning has no effect.}
\end{figure}
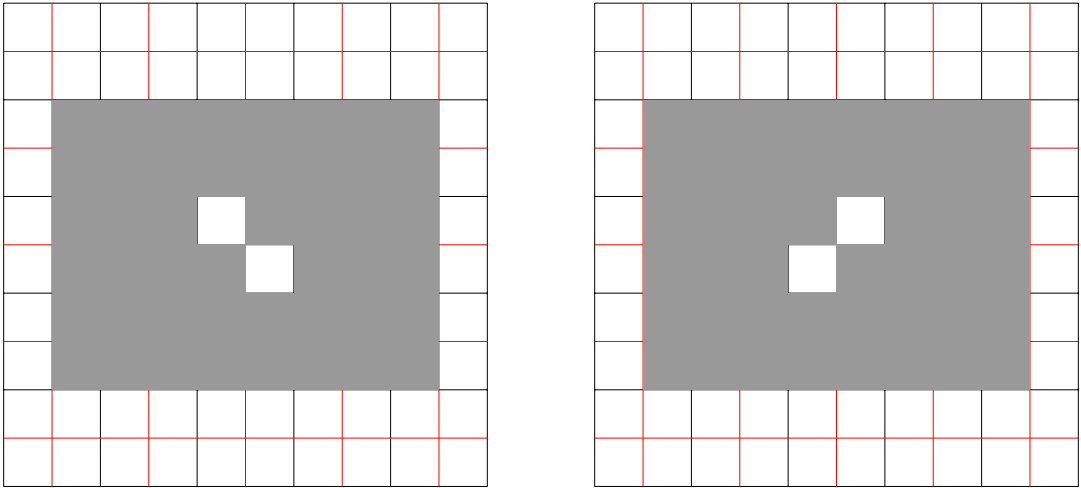

\section{Cost of operations in physically universal CAs\label{sec:physUniversal}}

\subsection{Schaeffer's physically universal CA}

Schaeffer \cite{Schaeffer2015}, see also
\cite{SchaefferWebsite}, constructed an
efficiently physically universal CA that is close to Margolus' billiard ball model CA. 
The update rules are shown in Figure~\ref{fig:schaeffer}. Here, physical universality refers to the Moore CA whose update rule consists of
two time steps of the Margolus CA
(following the remarks in Subsection~\ref{subsec:mbbm}).
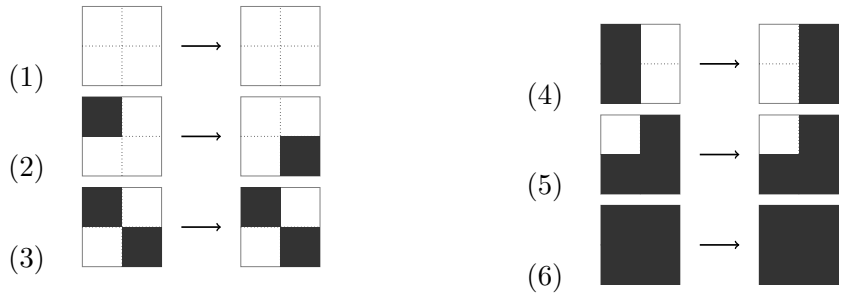
\begin{figure}
\centerline{
\begin{tabular}{cc}
(1) & 
\resizebox{0.2\textwidth}{!}{
\begin{tikzpicture}
\draw[step=1cm,gray,very thin] (0,0) rectangle (2,2);
\draw[dotted] (0,1) -- (2,1);
\draw[dotted] (1,0) -- (1,2);
\draw[very thick,->] (2.5,1) -- (3.5,1);
\draw[step=1cm,gray,very thin] (4,0) rectangle (6,2);
\draw[dotted] (4,1) -- (6,1);
\draw[dotted] (5,0) -- (5,2);
\end{tikzpicture}
}
\\
(2) & 
\resizebox{0.2\textwidth}{!}{
\begin{tikzpicture}
\draw[step=1cm,gray,very thin] (0,0) rectangle (2,2);
\draw[dotted] (0,1) -- (2,1);
\draw[dotted] (1,0) -- (1,2);
\fill[black!80!white] (0,1) rectangle (1,2);
\draw[very thick,->] (2.5,1) -- (3.5,1);
\draw[step=1cm,gray,very thin] (4,0) rectangle (6,2);
\draw[dotted] (4,1) -- (6,1);
\draw[dotted] (5,0) -- (5,2);
\fill[black!80!white] (5,0) rectangle (6,1);
\end{tikzpicture}
}
\\
(3) & 
\resizebox{0.2\textwidth}{!}{
\begin{tikzpicture}
\draw[step=1cm,gray,very thin] (0,0) rectangle (2,2);
\draw[dotted] (0,1) -- (2,1);
\draw[dotted] (1,0) -- (1,2);
\fill[black!80!white] (0,1) rectangle (1,2);
\fill[black!80!white] (1,0) rectangle (2,1);
\draw[very thick,->] (2.5,1) -- (3.5,1);
\draw[step=1cm,gray,very thin] (4,0) rectangle (6,2);
\draw[dotted] (4,1) -- (6,1);
\draw[dotted] (5,0) -- (5,2);
\fill[black!80!white] (4,1) rectangle (5,2);
\fill[black!80!white] (5,0) rectangle (6,1);
\end{tikzpicture}
}
\end{tabular}
\hspace{2cm}
\begin{tabular}{cc}
\\
(4) & 
\resizebox{0.2\textwidth}{!}{
\begin{tikzpicture}
\draw[step=1cm,gray,very thin] (0,0) rectangle (2,2);
\draw[dotted] (0,1) -- (2,1);
\draw[dotted] (1,0) -- (1,2);
\fill[black!80!white] (0,1) rectangle (1,2);
\fill[black!80!white] (0,0) rectangle (1,1);
\draw[very thick,->] (2.5,1) -- (3.5,1);
\draw[step=1cm,gray,very thin] (4,0) rectangle (6,2);
\draw[dotted] (4,1) -- (6,1);
\draw[dotted] (5,0) -- (5,2);
\fill[black!80!white] (5,0) rectangle (6,1);
\fill[black!80!white] (5,1) rectangle (6,2);
\end{tikzpicture}
}
\\
(5) & 
\resizebox{0.2\textwidth}{!}{
\begin{tikzpicture}
\draw[step=1cm,gray,very thin] (0,0) rectangle (2,2);
\draw[dotted] (0,1) -- (2,1);
\draw[dotted] (1,0) -- (1,2);
\fill[black!80!white] (1,1) rectangle (2,2);
\fill[black!80!white] (1,0) rectangle (2,1);
\fill[black!80!white] (0,0) rectangle (1,1);
\draw[very thick,->] (2.5,1) -- (3.5,1);
\draw[step=1cm,gray,very thin] (4,0) rectangle (6,2);
\draw[dotted] (4,1) -- (6,1);
\draw[dotted] (5,0) -- (5,2);
\fill[black!80!white] (4,0) rectangle (5,1);
\fill[black!80!white] (5,1) rectangle (6,2);
\fill[black!80!white] (5,0) rectangle (6,1);
\end{tikzpicture}
}
\\
(6) & 
\resizebox{0.2\textwidth}{!}{
\begin{tikzpicture}
\draw[step=1cm,gray,very thin] (0,0) rectangle (2,2);
\draw[dotted] (0,1) -- (2,1);
\draw[dotted] (1,0) -- (1,2);
\fill[black!80!white] (1,1) rectangle (2,2);
\fill[black!80!white] (1,0) rectangle (2,1);
\fill[black!80!white] (0,0) rectangle (1,1);
\fill[black!80!white] (0,1) rectangle (1,2);
\draw[very thick,->] (2.5,1) -- (3.5,1);
\draw[step=1cm,gray,very thin] (4,0) rectangle (6,2);
\draw[dotted] (4,1) -- (6,1);
\draw[dotted] (5,0) -- (5,2);
\fill[black!80!white] (4,0) rectangle (5,1);
\fill[black!80!white] (5,1) rectangle (6,2);
\fill[black!80!white] (5,0) rectangle (6,1);
\fill[black!80!white] (4,1) rectangle (5,2);
\end{tikzpicture}
}
\end{tabular}
}
\caption{\label{fig:schaeffer} Transition rules of 
Schaeffer's physically universal CA. Further rules are given by rotation invariance.}
\end{figure}
We now discuss a rather primitive solution of
implementing our bit switching task
in Schaeffer's CA. Its resource requirements 
grow at least linearly in $n$, which at first appears to be suboptimal. 
Yet, we will later show that linear growth is optimal.
We first observe that
the CA admits free particle propagation in diagonal direction, a
fact that is heavily used in the proof
for physical universality \cite{Schaeffer2015}. Figure~\ref{fig:particleprop} visualizes
this motion. 
\begin{figure}
\centerline{
\resizebox{0.8\textwidth}{!}{
\begin{tikzpicture}
\draw[step=2cm,black,thin] (0,0) grid (4,4);
\draw[red,thin] (1,0) -- (1,4);
\draw[red,thin] (3,0) -- (3,4);
\draw[red,thin] (0,1) -- (4,1);
\draw[red,thin] (0,3) -- (4,3);
\fill[black!80!white] (1,2) rectangle (2,3);
\draw[very thick,->] (5.5,2) -- (6.5,2);
\draw[step=2cm,black,thin] (8,0) grid (12,4);
\draw[red,thin] (9,0) -- (9,4);
\draw[red,thin] (11,0) -- (11,4);
\draw[red,thin] (8,1) -- (12,1);
\draw[red,thin] (8,3) -- (12,3);
\fill[black!80!white] (10,1) rectangle (11,2);
\draw[very thick,->] (13.5,2) -- (14.5,2);
\draw[step=2cm,black,thin] (16,0) grid (20,4);
\draw[red,thin] (17,0) -- (17,4);
\draw[red,thin] (19,0) -- (19,4);
\draw[red,thin] (16,1) -- (20,1);
\draw[red,thin] (16,3) -- (20,3);
\fill[black!80!white] (19,0) rectangle (20,1);
\end{tikzpicture}
}
}
\caption{\label{fig:particleprop} Free particle
propagation in Schaeffer's physically universal CA: the configuration on the left turns into the one in the middle by applying Rule (2) to the red partitioning. The middle configuration turns into the right one by applying the same rule to the black partitioning.}
\end{figure}
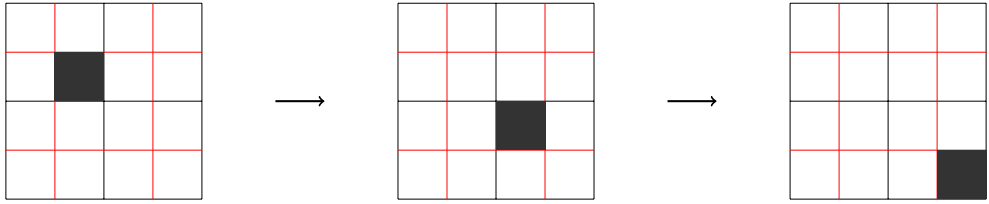
We now use a `beam of particles' in diagonal direction in which
a particle and a hole alternate, as shown in Figure~\ref{fig:beam}.
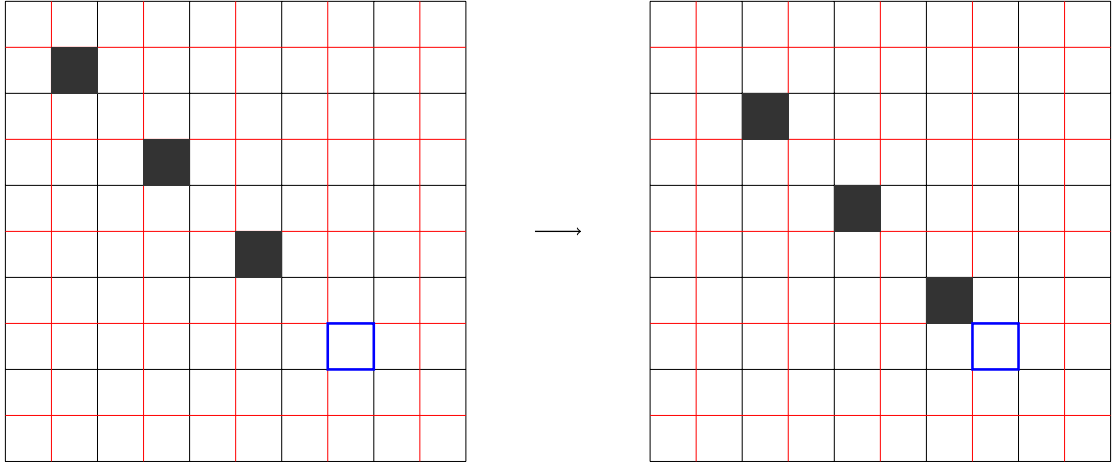
\begin{figure}
\centerline{
\resizebox{0.9\textwidth}{!}{
\begin{tikzpicture}
\draw[step=2cm,black,thin] (0,0) grid (10,10);
\draw[red,thin] (1,0) -- (1,10);
\draw[red,thin] (3,0) -- (3,10);
\draw[red,thin] (5,0) -- (5,10);
\draw[red,thin] (7,0) -- (7,10);
\draw[red,thin] (9,0) -- (9,10);
\draw[red,thin] (0,1) -- (10,1);
\draw[red,thin] (0,3) -- (10,3);
\draw[red,thin] (0,5) -- (10,5);
\draw[red,thin] (0,7) -- (10,7);
\draw[red,thin] (0,9) -- (10,9);
\fill[black!80!white] (1,8) rectangle (2,9);
\fill[black!80!white] (3,6) rectangle (4,7);
\fill[black!80!white] (5,4) rectangle (6,5);
\draw[blue,ultra thick] (7,2) rectangle (8,3);
\draw[thick,->] (11.5,5) -- (12.5,5);
\draw[step=2cm,black,thin] (14,0) grid (24,10);
\draw[red,thin] (15,0) -- (15,10);
\draw[red,thin] (17,0) -- (17,10);
\draw[red,thin] (19,0) -- (19,10);
\draw[red,thin] (21,0) -- (21,10);
\draw[red,thin] (23,0) -- (23,10);
\draw[red,thin] (14,1) -- (24,1);
\draw[red,thin] (14,3) -- (24,3);
\draw[red,thin] (14,5) -- (24,5);
\draw[red,thin] (14,7) -- (24,7);
\draw[red,thin] (14,9) -- (24,9);
\fill[black!80!white] (16,7) rectangle (17,8);
\fill[black!80!white] (18,5) rectangle (19,6);
\fill[black!80!white] (20,3) rectangle (21,4);
\draw[blue,ultra thick] (21,2) rectangle (22,3);
\end{tikzpicture}
}
}
\caption{\label{fig:beam} Beam of $3$ propagating particles which implement the turning on and off of the blue target cell $3$ times.}
\end{figure}
Then choose a target bit along the diagonal, as indicated by the
blue square in Figure~\ref{fig:beam}. Just by waiting, this bit is turned on and off when particles and holes appear, respectively.
The resource requirements of this implementation are large:
not only does it require to correctly locate particles and holes, it also requires to keep the space around the beam empty to protect the beam from collisions. 

\begin{Remark}[complexity aspect of preparation]
Apart from being costly from the thermodynamic point of view, 
the implementation is also `not nice' in other respects:
compared to the simplicity of our control problem,
the initialization is rather complex.
Assume, for comparison, the following general control task:
given some arbitrary binary string $b$ of length $2n$, 
the target bit is supposed to attain the value $b_j$ at ime $j$.
Then, the above beam solves this task for the special case where 
$b=101010\cdots 10$. The general task can be obviously solved by the same procedure as above: just locate particles and holes according to $b$. The fact that the solution of the simple special case
is based on the same principle suggests that it is a `bad' solution; it is inappropriately complex compared to the simplicity of the task. In a way, it reduces a simple control operation to
one that seems more complex. This raises the question of what
one  wants to call a `solution' of a control task.
\end{Remark}

To return to the thermodynamic question, one may wonder if there exist smarter implementations of the bit switch process where the resource requirements do not grow linearly in $n$.
We can easily show that the {\it range} of the implementation of
the $n$-fold bit switch grows linearly in $n$.
To this end, we
first need the Diffusion Theorem of \cite{Schaeffer2015}:

\begin{Theorem}[Diffusion Theorem]\label{thm:diffusion}
Let $S\subset \Z^2$ be an arbitrary square of side length $s$ in the Moore CA and $\phi$ an arbitrary configuration that is empty on $\bar{S}$.  Then $\alpha^t(\phi)$ is empty on $S$ for all $t\ge s$.
\end{Theorem}
We then have:
\begin{Theorem}[range of device restoring a region after $t$ time steps]\label{thm:storage}
Let $f : \Sigma^X \rightarrow \Sigma^X$ denote an arbitrary bijection for some region $X\subset \Z^2$. Assume that  
$(Z,\phi_Z,t)$ is a device for implementing $f$. Then its range is at least $t$.
\end{Theorem}

\begin{proof}
Let $S$ be the smallest square containing $W=Z\cup X$ and $s$ its side length.  We must have $s\ge t$. Assume to the contrary that $s\le t-1$.  Then by the diffusion theorem any configuration $\phi$ that is empty on $\bar{S}$ evolves in $t$ times steps to a configuration that is empty on $S$.  This contradicts the assumption there there is a
configuration $\phi_Z$ that implements a bijection on site $X$. 
\end{proof}

\noindent
An important special case of the above theorem is when $f$ is the identity function ${\rm ID}$.
Moreover, we have:

\begin{Theorem}[range of device implementing $n$ powers of a transformation]\label{thm:repetitions}
Let $f: \Sigma^X \to \Sigma^X$ be an arbitrary bijection for some region $X \subset \Z^2$.
Assume that $(Z,\phi_Z,t_1,t_2,\dots,t_n)$ 
be a device for implementing of the sequence $f,f,\dots,f$.
Then its range is at least $n$.
\end{Theorem}

\begin{proof}
The proof is very similar to the proof of Theorem~\ref{thm:storage}.  If $W=Z\cup X$ were contained in a square
of side length $s \le t_n-1$ the configuration after $t_n$ would be empty on $X$. Thus $s\ge t_n$. The result 
follows since we must have $t_n\ge n$. 
\end{proof}

\begin{Remark}[resources requirements for 1D physically universal CA]
We make some comments on resource requirements of the one-dimensional physically universal CA in \cite{Salo2015}.
This CA uses interacting particles particles that propagate with different speeds, namely $\pm 1$ or $\pm 2$ sites per time step. Similar results to 
Theorem~\ref{thm:storage} and Theorem~\ref{thm:repetitions} hold for this CA as well.

\cite[Lemma 2]{Salo2015} is also a kind of diffusion theorem similar to Theorem~\ref{thm:diffusion}. 
We reformulate its statements slightly. Let $S$ be an interval of length $s$ and $\phi$ a configuration that is empty on $\bar{S}$. Then, after $t(s)\in O(s)$ time steps all configurations $\phi'$ that arise from $\phi$ under the autonomous time evolution are empty on $S$. It is convenient to rephrase $t(s)\in O(s)$ as follows: there exist two contants $s_0$ and $\kappa$ such that $t(s)\le \kappa s$ for all $s\ge s_0$.

Using the same arguments but now with the one-dimensional diffusion theorem, we may conclude that for the one-dimensional CA the ranges must be at least $t/\kappa$ and $n/\kappa$ in Theorem~\ref{thm:storage} and Theorem~\ref{thm:repetitions}, respectively, provided that $X$ is sufficiently large. The latter condition on $X$ is necessary because the diffusion theorem only applies for intervals of length at least $s_0$.

\end{Remark}

The range is a rather crude measure of the resource requirements. A finer measure is the size, that is, the number of cells of the relevant region.  We focus the elementary  control task of restoring a bit $n$ times and derive a lower bound on the size of the corresponding device.

\begin{Theorem}[size of device restoring a bit $n$ times]\label{thm:main}
Let ${\rm ID}$ denote the identity on some cell of $\Z^2$ in the Moore CA corresponding to Schaeffer's construction. Assume that $(Z,\phi_Z,t_1,t_2,\dots,t_n)$ is a device for implementing ${\rm ID},{\rm ID},\dots,{\rm ID}$. Then
$Z$ contains at least  $n/4-1$ cells (also counted in the Moore CA).
\end{Theorem}

\begin{proof}
Below, the term 'cell' refers to
a cell in the Margolus CA (containing just one bit), not the $2\times 2$ block defining the cell of the corresponding Moore CA.
Let $X$ denote the source/target $2\times 2$ block. 
Since $Z$ consists, by definition, of cells 
of the Moore CA, it consists only of complete 
$2\times 2$ blocks in the Margolus CA.

We now rely on the techniques developed
in the proof of Theorem~4 in \cite{Schaeffer2015}.
We also consider
an `abstract' CA that consists of three states $\{0,1,\top\}$, where
$\top$ denotes a `wild card' that stands for an uncertain state. The 
purpose of the abstract CA is merely to keep track of how 
uncertain states propagate in the concrete CA. Ref.~\cite{Schaeffer2015}
describes a rather simple set of update rules for the abstract CA, whose details are not needed.
The essential observation that we adopt is that $\top$ particles exhibit free particle propagation as
long as the following `forbidden' patterns

\vspace{0.2cm}
\centerline{
\resizebox{0.2\textwidth}{!}
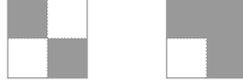
{
\begin{tikzpicture}
\draw[step=1cm,gray,very thin] (0,0) rectangle (2,2);
\draw[dotted] (0,1) -- (2,1);
\draw[dotted] (1,0) -- (1,2);
\fill[black!40!white] (0,1) rectangle (1,2);
\fill[black!40!white] (1,0) rectangle (2,1);
\draw[step=1cm,gray,very thin] (4,0) rectangle (6,2);
\draw[dotted] (4,1) -- (6,1);
\draw[dotted] (5,0) -- (5,2);
\fill[black!40!white] (4,1) rectangle (5,2);
\fill[black!40!white] (5,1) rectangle (6,2);
\fill[black!40!white] (5,0) rectangle (6,1);
\end{tikzpicture}
}
}

\vspace{0.2cm}
\noindent
and their rotated versions do  not occur. Here a grey box indicates the $\top$ state, which stands for either the $0$ state (white) or the $1$ state (black).
It is important that these forbidden patterns will never occur during the dynamical evolution of the abstract CA if the initial configuration does not contain any forbidden patterns \cite{Schaeffer2015}.

First, we assign $\top$ to all cells in 
$W:=Z+X$
representing the fact that
their states are unknown or could be arbitrary (because we do not know what the correct $\phi_Z$ looks like and the source block $X$ could also be in any state). Second, we assign 
$0$ to all cells in the complement of 
$W$. 
This is possible because cells outside 
$W$ 
do not matter for the correct implementation.

This way, the forbidden patterns do not occur in the initial configuration and the dynamics of the abstract CA can be described
by free propagation of $\top$-particles: each $\top$-particle moves to the diagonally opposite cell, that is, in either NE, NW, SE, SW direction.
Consequently, any cell can attain $\top$ and, in particular $1$, at most 
$4|W|$ 
times.   

Assume one of the cells in the source region $X$ is in the state $1$ at $t=0$. Consequently, it must be in the state $1$ at least $n$ times during the interval $1,2,\ldots,t_n$ to ensure the correct implementation of the $n$-fold repetition of ${\rm ID}$. By  combining these two arguments together we conclude that
$W$ consists of at least $n$ cells of the Margolus CA. Hence, it consists of at least
$n/4$ cells of the Moore CA. Since
$Z$ differs from $W$ by only one cell
we finally obtain the lower bound $|Z| \geq n/4 -1$.
\end{proof}

\noindent
Theorem~\ref{thm:main} can easily be applied to our task of
$n$-fold NOT since the latter amounts to implementing the identity for
all $t_j$ with even $j$. 
 
It is unclear whether some of these insights apply to a general physically universal CA.
The question whether there exist physically universal CAs that do not satisfy the Diffusion Theorem has already been raised by Schaeffer \cite{Schaeffer2015}, which seems related to our thermodynamic questions since diffusion is what makes information so 
extremely unstable.

It is, however, clear that
in any physically universal CA a configuration of a finite region is unstable in the following sense:
\begin{Theorem}[instability of patterns]\label{thm:instability}
For some  physically universal CA, let $Z\subset \Z^d$ be a finite region that is initialized to the state
$\phi_Z$. Assume that the states of all cells of $\bar{Z}$ are unknow and described by some probability distribution $P$ that assigns 
a non-zero probability to every possible state in $\Sigma^{\bar{Z}}$. Then, for any configuration
$\phi'_Z$ of $Z$ there is a time $t$ such that $\phi_Z$ evolves to $\phi_Z'$ with non-zero probability.
\end{Theorem}
\begin{proof}
Choose a function $f:\Sigma^Z\to \Sigma^Z$ with $f(\phi_Z)=\phi_Z'$.
By physical universality, there is a configuration of the complement of $Z$
implementing $f$ for some $t$. Since only the restriction of the configuration to a finite region 
matters (cells that are further away than $t$ sites do not matter) the set of all configurations
implementing $f$ has a non-zero probability.
\end{proof}

\noindent
The absence of impenetrable walls in physically universal automata is only the most
intuitive consequence of this obervation. Less obvious consequences remain to be discovered in the future.

\section{Conclusions}

Common discussions on thermodynamic irreversibility of operations often focus on
entropy generation while they substantially differ with respect to the underlying notion of entropy (e.g. Boltzmann entropy, Shannon respective von Neumann entropy, or Kolmogorov complexity 
\cite{Bennett:89,ZurekKol,AICarrowoftime}). 
Given these different notion of entropy, entropy generation is explained 
by coarse graining \cite{Balian1}, because complexity also contributes to physical entropy by definition \cite{Bennett:89,ZurekKol}, or because entropy leaks into the system from its environment. 

Irreversibility in physically universal CAs or Hamiltonian
systems is not due to entropy production -- at least not in any obvious sense. Instead, every evolution is to some extent irreversible simply because one has no access to the evolution, the autonomous time evolution of the system just 
continues forever. Therefore, simulating the inverse evolution on some target system 
involves sophisticated initialization of a large number of
cells in the surrounding (acting as the controller).
Since this initialization is typically destroyed by the autonomous evolution of the system,  restoring the state of the joint system of target {\it and} its controller involves
a sophisticated initialization of a `meta-controller', which, in turn,
will then be destroyed by the evolution. The question of how to reverse the dynamics of one system without disturbing the state of its surrounding thus raises the same question for an even larger system.

The idea that control operations, even when they are unitary,
imply heat generation in the controlling device, 
is certainly not new.  However,
physically universal CAs and Hamiltonians
may allow us to look at the idea from a new perspective
because they 
admit to describe target, controller, meta-controller and so on,
in a unified way since all of them are just regions of cells.
Moreover, physically universal CAs formalize the conflict between
controllability and isolability of a system in a  principled
way. This is because  physical universality,
which formalizes the ability to control subsystems,
implies instability of information, although quantitative results have to be left to the future. 
Here we have shown that in the existing constructions of physically universal cellular automata information is extremely unstable -- for instance, in the sense that the resource required for protecting information grows linearly in time.

The intention of this article is to inspire other researchers to explore implications of physical universality rather than exploring properties of specific constructions of CAs. Here we have discussed properties of Schaeffer's construction only to illustrate how to work with
our notion of resource requirements in the context of a physically universal CA.

\paragraph{Acknowledgements:} We would like to thank Scott Aaronson and Luke Schaeffer for 
interesting discussions on related questions and Armen Allahverdyan for comments on an earlier version of this manuscript.


\end{document}